\documentclass[hidelinks, 12pt]{article}
\usepackage[english]{babel}
\usepackage{amssymb,amsmath,amsthm,amsfonts,mathtools}
\usepackage{mathabx}
\usepackage[top=3.5cm, bottom=3.5cm, left=3.3cm, right=3.3cm]{geometry}
\usepackage{graphicx}
\usepackage{amsfonts}
\usepackage[utf8]{inputenc}
\usepackage{tikz}
\usepackage{tikz-cd}
\usetikzlibrary{trees}
\usepackage{verbatim}
\usepackage{appendix}
\usepackage{natbib}
\usepackage{datetime}
\usepackage{setspace}
\usepackage{bookmark}
\usepackage{hyperref}

\newtheorem{prop}{Proposition}
\theoremstyle{definition}

\theoremstyle{proof}

\theoremstyle{definition}

\numberwithin{equation}{section}
\renewenvironment{abstract}
 {\small
  \begin{center}
  \bfseries \abstractname\vspace{-.5em}\vspace{0pt}
  \end{center}
  \list{}{
    \setlength{\leftmargin}{.05cm}%
    \setlength{\rightmargin}{\leftmargin}%
  }%
  \item\relax}
 {\endlist}

\begin{document}

\onehalfspacing


\title{Optimin achieves super-Nash performance}
\author{Mehmet S. Ismail\footnote{ {\footnotesize Department of Political Economy, King's College London, London, UK. E-mail: mehmet.ismail@kcl.ac.uk}}}
\date{\today}
\maketitle
\begin{abstract}

Since the 1990s, AI systems have achieved superhuman performance in major zero-sum games where ``winning'' has an unambiguous definition. However, most social interactions are mixed-motive games, where measuring the performance of AI systems is a non-trivial task. In this paper, I propose a novel benchmark called \textit{super-Nash performance} to assess the performance of AI systems in mixed-motive settings.  I show that a solution concept called optimin achieves super-Nash performance in every $n$-person game, i.e., for every Nash equilibrium there exists an optimin where every player not only receives but also \textit{guarantees} super-Nash payoffs even if the others deviate unilaterally and profitably from the optimin.

\end{abstract}

\noindent \textit{Keywords}: maximin criterion, noncooperative games, cooperative games, Nash equilibrium, traveler's dilemma, repeated prisoner's dilemma
\newpage
\section{Introduction and motivating examples}
\label{sec:intro}

Since the early 1990s, artificial intelligence (AI) systems have gradually achieved superhuman performance in major competitive games such as checkers, chess, Go, and poker \citep{campbell2002,schaeffer2007,silver2016,brown2019}. All these games belong to a class of zero-sum games in which the notion of ``winning'' is clearly defined. However, most social interactions are mixed-motive (general-sum) games, and measuring the performance of AI systems in this setting is a non-trivial task. 

In this paper, I define a novel benchmark in $n$-person games: a strategy profile played in a game is said to achieve \textit{super-Nash performance} if each player (i) receives a super-Nash payoff and (ii) guarantees a super-Nash payoff under any unilateral profitable deviation by the other players. I next justify this definition with a theoretical and an empirical argument. First, a ``good'' measure should not be based solely on the external payoffs (i.e., rewards) received by the AIs but should also take into account the counterfactual payoffs---i.e., how well these AIs would perform if they had to play against an AI that ``exploited'' their strategies. For example, two AIs that always play cooperate in a prisoner's dilemma would receive a strictly greater payoff than their Nash equilibrium payoff. However, it is clear that they would perform poorly against an opportunistic AI that plays defect.  Second, overwhelming experimental evidence suggests that humans already achieve super-Nash performance in a variety of general-sum games in which there are gains from cooperation. These include the finitely repeated prisoner's dilemma, the finitely repeated public goods game, the centipede game, and the traveler's dilemma \citep{axelrod1980,mckelvey1992,capra1999,goeree2001,rubinstein2007,lugovskyy2017,embrey2017}. In these games, not only do humans consistently receive super-Nash payoffs, but they also guarantee these super-Nash payoffs even when other players anticipate non-Nash behavior and respond ``opportunistically.''

I prove that a solution concept called optimin \citep{ismail2020} achieves super-Nash performance in every $n$-person game. Informally, an optimin is a strategy profile in which each player maximizes their minimal payoff under unilateral profitable deviations of the others. Optimin achieves super-Nash performance in the sense that for every Nash equilibrium there exists an optimin where each player not only receives a super-Nash payoff but also \textit{guarantees} a super-Nash payoff even if the others deviate unilaterally and profitably from the optimin. To the best of my knowledge, optimin is the first solution concept that ensures super-Nash payoffs under profitable deviations of the others. As is well-known, a profile of maximin strategies can never Pareto dominate a Nash equilibrium. In a similar vein, a Nash equilibrium can never Pareto dominate an optimin point.

Furthermore, in any finite $n$-person game, there is an optimin point in pure strategies when the game is restricted to the pure strategies. In contrast, a Nash equilibrium in pure strategies does not exist in general. Moreover, computing mixed Nash equilibria is ''hard'' \citep{gilboa1989b,daskalakis2009}. This is not to argue that computing mixed optimin is more efficient. However, one advantage of pure strategy optimin is that it is sufficient to restrict attention to the set of pure strategy profiles (a finite set) compared to mixed strategy profiles (an infinite set). This property of pure optimin can be particularly valuable for finding a solution in large games.

\subsection{Illustrative example} 
\begin{figure}
	\[
	\begin{array}{ r|c|c|c| }
	\multicolumn{1}{r}{}
	&  \multicolumn{1}{c}{\text{Left}}
	& \multicolumn{1}{c}{\text{Center}}
	& \multicolumn{1}{c}{\text{Right}} \\
	\cline{2-4}
	\text{Top} &  100,100 & 100,105 & 0,0 \\
	\cline{2-4}
	\text{Middle}&  105,100 & 95,95 & 0,210  \\
	\cline{2-4}
	\text{Bottom}&  0,0 & 210,0 & 5,5\\
	\cline{2-4}
	\end{array}
	\qquad
	\begin{array}{ r|c|c|c| }
	\multicolumn{1}{r}{}
	&  \multicolumn{1}{c}{\text{Left}}
	& \multicolumn{1}{c}{\text{Center}}
	& \multicolumn{1}{c}{\text{Right}} \\
	\cline{2-4}
	\text{Top} &  (100,100) & (100,0) & (0,0) \\
	\cline{2-4}
	\text{Middle}&  (0,100) & (0,0) & (0,5)  \\
	\cline{2-4}
	\text{Bottom}&  (0,0) & (5,0) & (5,5)\\
	\cline{2-4}
	\end{array}
	\]
	\caption{An illustrative game (left) and its minimal payoffs (right). The unique optimin point is (Top, Left), whereas every strategy is a maximin strategy in this game. The unique Nash equilibrium is (Bottom, Right).}
	\label{fig:illustrative}
\end{figure}

To illustrate the optimin criterion, consider the game in Figure~\ref{fig:illustrative} (left) in which, for the sake of simplicity, the attention is restricted to pure strategies.\footnote{This example is taken from \citet{ismail2020}.} Consider strategy profile (Top, Left). Player 2 (he) has a profitable deviation to `Center' but player 1 (she), who plays Top, would still receive 100 at profile (Top, Center). Player 2 might deviate to `Right' from (Top, Left), though such a deviation is not plausible because he would incur a huge loss.\footnote{Note that this deviation would be allowed under the maximin strategy concept \citep{neumann1928,neumann1944}.} Therefore, each player's minimal payoff associated with (Top, Left) is 100 given the unilateral profitable deviation of the opponent. As illustrated in Figure~\ref{fig:illustrative} (right), (Top, Left) is the unique optimin point, because it maximizes the minimal payoffs.

\section{Super-Nash performance} 
Let $G=(\Delta X_i, u_i)_{i\in N}$ be a finite $n$-person non-cooperative game in mixed extension where a mixed strategy profile is denoted by $p\in\Delta X$.

A solution concept is said to achieve \textit{super-Nash performance} if for every $n$-person game $G$ and for every Nash equilibrium $r\in \Delta X$ in $G$, it gives a solution $q\in \Delta X$ such that for every player $i$ 
\[
\inf_{p'_{-i}\in B_{-i}(q)} u_{i}(q_i,p'_{-i}) \geq  u_{i}(r),
\]
where $B_i(p)\coloneqq\{p'_i\in \Delta X_i| u_i(p'_i,p_{-i})>u_i(p_i,p_{-i})\}\cup \{p_i\}$ and $B_{-i}(p)\coloneqq\bigtimes_{j\in N\setminus\{i\}}B_{j}(p)$. In addition, there exists a game $G'$ where the above inequality is strict for every $i$.

In plain words, a solution concept achieves super-Nash performance if for every game and for every Nash equilibrium in the game, the concept outputs a mixed strategy profile such that each player guarantees at least their Nash equilibrium payoff in the game under unilateral profitable deviations by the other players. In addition, there must at least be a game in which each player guarantees strictly more payoff than their Nash equilibrium payoff.

A mixed strategy profile $p^*\in \Delta X$ is said to satisfy the \textit{optimin criterion}, or called an \textit{optimin point}, if for every player $i$, $p^*$ solves the following multi-objective optimization problem \citep{ismail2020}. 
\[
p^* \in \arg \max_{q\in \Delta X} \inf_{p'_{-i}\in B_{-i}(q)} u_{i}(q_i,p'_{-i}).
\]

In plain words, a strategy profile $p^*$ is called an optimin if each player simultaneously maximizes their minimal payoff under the unilateral profitable deviations by other players.

The main difference between the optimin and the maximin is that the optimin restricts players to only profitable deviations from a strategy profile. Under the maximin concept, the potential ``performance'' of a strategy profile is assessed by its minimum payout under any deviation by the others, whereas under the optimin it is measured by its minimum payoff under unilateral and profitable deviations.

\section{Main result} 

Next, I show that for every Nash equilibrium there is an optimin point in which not only is every player (weakly) better off but also \textit{guarantees} to be better off even if the other players deviate unilaterally and profitably from the optimin.\footnote{\cite{ismail2019} proves that in every game there exists an optimin point in possibly mixed strategies.}

\begin{prop}[Super-Nash performance]
\label{prop:nashvsoptimin}
For every Nash equilibrium $q^*$ there exists an optimin point $p^*$ such that for every player $i$ $\inf_{p'_{-i}\in B_{-i}(p^*)} u_{i}(p^*_i,p'_{-i})\geq u_i(q^*)$. Moreover, a Nash equilibrium $q^*$ can never Pareto dominate an optimin point $p^*$.
\end{prop}

\begin{proof}
First, notice that if $q^*$ is an optimin point, then we are done. So I assume that $q^*$ is not an optimin point. To reach a contradiction, suppose that there exists a Nash equilibrium $q^*$ for every optimin $p^*$ there exists a player $j$ such that 
\[
\inf_{p'_{-j}\in B_{-j}(p^*)} u_{j}(p^*_j,p'_{-j})< u_j(q^*).
\]
This implies that for every $i$ $q^* \in \arg \max_{q\in \Delta X} \inf_{p'_{-i}\in B_{-i}(q)} u_{i}(q_i,p'_{-i})$ because for every player $i$ $\inf_{p'_{-i}\in B_{-i}(q^*)} u_{i}(q^*_i,p'_{-i}) = u_i(q^*)$. This means that $q^*$ is an optimin point, which contradicts to our supposition that it is not.

Next, I show that a Nash equilibrium $q^*$ cannot Pareto dominate an optimin point $p^*$. To reach a contradiction suppose that $q^*$ Pareto dominates $p^*$. Note that for every player $j$ $\inf_{p'_{-j}\in B_{-j}(q^*)} u_{j}(q^*_j,p'_{-j}) = u_j(q^*)$ and $u_j(q^*) \geq u_j(p^*)$ (with at least one strict inequality) by our supposition. However, then $p^*$ is not an optimin point, which contradicts our supposition.
\end{proof}

Note that Proposition~\ref{prop:nashvsoptimin} would remain true if we interchange Nash equilibrium with a profile of maximin strategies in an $n$-person game. This is because it is a well-known fact that a player's payoff from a Nash equilibrium cannot be strictly lower than their maximin strategy payoff.

It is also well-known that for every Nash equilibrium there is a Pareto efficient profile in which every player is (weakly) better off under the efficient profile. However, if a player profitably deviates from an efficient profile---like deviating from (C,C) in the prisoner's dilemma---it can be disastrous for the non-deviators. What is remarkable with optimin is that by Proposition~\ref{prop:nashvsoptimin} players guarantee super-Nash payoffs even if others do deviate unilaterally and profitably.

While Proposition~\ref{prop:nashvsoptimin} described above is a promising result from a technical point of view, it is not immediately clear how different optimin and Nash equilibrium predictions can be in experimental games. I next apply optimin to two well-studied experimental games in which Nash equilibrium and optimin predictions are in stark contrast. The optimin predictions are consistent with the comparative statics of human behavior in these games, whereas it is well-established that Nash equilibrium predictions are not.

\subsection{The finitely repeated prisoner's dilemma} 

The finitely repeated prisoner's dilemma is a game in which defection is a dominant strategy in the stage game (but not in the repeated game), which is illustrated below. 
\[
\begin{array}{ r|c|c| }
	\multicolumn{1}{r}{}
	&  \multicolumn{1}{c}{\text{Cooperate}}
	& \multicolumn{1}{c}{ \text{$\;$ Defect $\;$} }\\
	\cline{2-3}
	\text{Cooperate}&  3,3 & 0,5 \\
	\cline{2-3}
	\text{Defect}&  5,0 & 1,1 \\
	\cline{2-3}
\end{array}
\]
\vspace{0.5 mm}

The unique subgame perfect equilibrium, as is widely known, prescribes players to defect in every round. Experiments show that (i) initial cooperation grows as the number of rounds increases, and (ii) cooperation decreases as the game progresses; for a meta analysis of many prisoner's dilemma experiments, see \cite{embrey2017} and the references therein. These regularities can be explained by the optimin criterion. Although the unique optimin point coincides with the unique Nash equilibrium in the one-shot game, the Tit-for-Tat profile generally satisfies the optimin criterion when the game is repeated. This is because even if a player attempts to exploit cooperative behavior, the cooperator's minimal payoff is in general greater than the subgame perfect equilibrium payoff.  The minimal payoffs of cooperation rise as the number of rounds increases in a game. However, within a repeated game, these minimal payoffs progressively decline and finally become smaller than the minimal payoffs of defection, which explains the declining amount of cooperation in the last rounds of a repeated game.

To illustrate why the Tit-for-Tat is an optimin, assume that the game is repeated 100 times. Because there is a profitable deviation in the last round, the Tit-for-Tat is evidently not a subgame perfect equilibrium. Nonetheless, by the last round in which a player deviates from the Tit-for-Tat profile, each player is already \textit{guaranteed} to receive a payoff of $99\times3=297$, which is nearly three times the payoff of a player in the unique subgame perfect equilibrium.

\subsection{The traveler's dilemma} 

\begin{figure}
	\[
	\begin{array}{ r|c|c|c|c|c| }
	\multicolumn{1}{r}{}
	&  \multicolumn{1}{c}{\text{100}}
	& \multicolumn{1}{c}{\text{99}}
	& \multicolumn{1}{c}{\cdots}
	& \multicolumn{1}{c}{\text{3}}
	& \multicolumn{1}{c}{\text{2}} \\
	\cline{2-6}
	\text{100}&  100,100 & 97,101 & \cdots & 1,5 & 0,4\\
	\cline{2-6}
	\text{99}&  101,97 & 99,99 & \cdots & 1,5 & 0,4\\
	\cline{2-6}
	\text{\vdots}&  \vdots & \vdots & \ddots & \vdots & \vdots \\
	\cline{2-6}
	\text{3}&  5,1 & 5,1 & \cdots & 3,3 & 0,4\\
	\cline{2-6}
	\text{2}&  4,0 & 4,0 & \cdots & 4,0 & 2,2 \\
	\cline{2-6}
	\end{array}
	\]
	\caption{Traveler's dilemma with reward/punishment parameter $r=2$.}
	\label{fig:travelers}
\end{figure}

Figure~\ref{fig:travelers} illustrates the traveler's dilemma in which each player selects a number between 2 and 100 \citep{basu1994}. The one who chooses the smaller number, $n$, receives $n+r$, where $r>1$, and the other player receives $n-r$. If they both choose $n$, then they each receive $n$. The reward parameter $r=2$ in the original game. Experimental research suggests that the reward/punishment parameter has a significant impact on the behavior. When $r$ is ``small,'' as in the original game, the subjects' behavior converges towards the highest number, whereas when $r$ is ``large,'' their behavior converges towards the lowest number \citep{capra1999,rubinstein2007}. The Nash equilibrium is (2,2) irrespective of $r$. In contrast, the optimin is consistent with the direction of play. The unique optimin point coincides with the Nash equilibrium when $r$ is large, but when $r$ is small, only the highest pair of numbers satisfies the optimin criterion. The reason for this is because as the reward parameter grows, so do the minimum payoffs of cooperation, and at some point, the minimal payoffs of the greatest number (100) become lower than the minimal payoffs of the lowest number (2).

\section{Conclusion}
\label{sec:conclusion}

To conclude, AI systems have achieved superhuman performance in several major zero-sum games since the early 1990s. However, most social interactions are mixed-motive games, where measuring the performance of AI systems is a non-trivial task. In this study, I have defined a benchmark notion of super-Nash performance to measure the performance of AI systems and humans in $n$-person mixed-motive settings. Accordingly, I have shown that in every $n$-person game there is always an optimin point that achieves this benchmark.

%
\bibliographystyle{chicago}

\bibliography{me}

\end{document}